\documentclass[11pt]{article}
\usepackage[T1]{fontenc}
\usepackage{lmodern}
\usepackage{amsmath,amssymb,amsthm}
\usepackage{geometry,booktabs,array,longtable,enumitem}
\usepackage{cite}
\usepackage{hyperref}
\usepackage{microtype}
\usepackage{needspace}
\hypersetup{hidelinks,pdftitle={The Residual 288 in an E8 x omega E8 Program: Representation Accounting and Physical Matter Interfaces},pdfauthor={Tejinder P. Singh}}
\newtheorem{proposition}{Proposition}

\newcommand{\st}{\mathrm{geom}}
\newcommand{\Tr}{\operatorname{Tr}}
\newcommand{\Cl}{\mathrm{Cl}}
\newcommand{\End}{\operatorname{End}}
\newcommand{\Hom}{\operatorname{Hom}}
\newcommand{\Lff}{\mathcal L_{FF}}

\title{The Residual $288$ in an $E_8\times\omega E_8$ Program:\\
Representation Accounting and Physical Matter Interfaces}
\author{Tejinder P. Singh\\
Tata Institute of Fundamental Research,\\
Homi Bhabha Road, Mumbai 400005, India\\
\texttt{tpsingh@tifr.res.in}}
\date{}
\begin{document}
\maketitle
\begin{abstract}
We examine the accounting identity $496=208+288$ in an octonionic $E_8\times\omega E_8$ unification program. The dimensions count selected representation labels, not a closed particle spectrum. In a consistent trinification convention, the quark-doublet labels descend from the $E_6$ adjoint, whereas the lepton-doublet labels descend from the mixed $E_8$ sectors. Their occurrence supplies useful nonabelian multiplicities but does not establish three complete chiral generations. We prove that the stated doublet assignments cannot acquire their Standard-Model hypercharges from a single $E_6$ generator commuting with the full family, colour and weak groups. Separately, a chosen charge-first Clifford benchmark carries three complete families on a $48$-dimensional internal Weyl coefficient space and admits an anomaly-free common family action. We construct explicit subgroup-equivariant lifts of physical scalar currents into $\overline{27}\otimes27=1\oplus78\oplus650$, and show that their nonuniqueness prevents a coupling prediction from normalization alone. The full comparison square contains fundamental family channels; a narrower exclusion applies only to undressed matter bilinears trivial under an independent geometric group. We identify the localization maps, interaction matrix elements, vacuum correlations and kinetic matrices needed to connect the exceptional geometry to physical matter. The geometric interpretation is retained, while generation selection, scalar binding and mass matching remain dynamical problems.
\end{abstract}

\section{Purpose and scope}\label{sec:intro}
The octonionic program of Kaushik, Vaibhav and Singh (KVS) uses two exceptional-algebra branches to organize proposed Standard-Model and pre-gravitational structures~\cite{KVS,STM,LRbiquaternion}. Generalized Trace Dynamics (GTD) and exceptional-Jordan constructions supply proposed dynamics and spectral mass data~\cite{GTDemergence,MassRatios}. Its geometric interpretation is retained here: the external $SU(3)$ in each $E_8\supset SU(3)\times E_6$ is geometric, and exceptional internal structures classify candidate interaction channels. Physical chiral fermions are specified separately as Clifford-based modules. The issue is whether explicit representations and interactions connect these layers.

The accounting $496=208+288$ compares the exceptional adjoints with a proposed physical roster. To assess what it establishes, one must ask two separate questions: which representation types occur, and which fields the dynamics actually propagate. A minimal Clifford ideal, a Jordan-algebra coordinate and an adjoint generator are different objects. None can be counted as a physical fermion merely because it has a familiar charge label. We therefore examine the accounting without assuming spectrum closure or a derivation of the family multiplicity.

\Needspace{7\baselineskip}
Three distinct mathematical objects will be kept separate:
\begin{enumerate}[label=(\roman*)]
\item the branched adjoints, used as a representation-label ledger;
\item a physical Standard-Model Weyl coefficient space, with its gauge action stated explicitly;
\item the exceptional comparison space $\End(27)$, which can organize physical operators only through a specified map.
\end{enumerate}
An intertwiner is a linear map that respects the stated symmetry actions: if $P$ maps one module to another, it must obey $P\rho_1(g)=\rho_2(g)P$. Matching dimensions or selected labels does not establish such a map. Nor does the presence of an internal representation establish Lorentz spin, chirality, a propagator or a light pole. These distinctions are compatible with division-algebraic approaches to matter~\cite{Baez2002,Furey,Furey2018EPJC,Stoica2018,Boyle2020}; they state the additional work needed for the present geometric interpretation to become predictive.

The argument proceeds in three stages. Section~\ref{sec:ledger} establishes the ledger and a restricted hypercharge obstruction. Section~\ref{sec:matter} states a physical representation that passes finite consistency checks but whose emergence is not derived. Sections~\ref{sec:bridge}--\ref{sec:dynamics} construct partial operator interfaces, delimit selection rules and specify the dynamical calculation still required. Appendix~\ref{app:verification} explains the reproducible checks.

\subsection{The microscopic seed and the accounting problem}
With $\mathcal D_B=\dot Q_B$, $\Psi_F=\dot Q_F$ and $\eta=L_P^2/L^2$, the single-aikyon expression used in the program is~\cite{GTDemergence}
\begin{equation}\label{eq:Lgtd}
\mathcal L_{\rm GTD}=\frac{\eta}{2}\Tr\!\left[
(\mathcal D_B^\dagger+\eta\beta_1\Psi_F^\dagger)
(\mathcal D_B+\eta\beta_2\Psi_F)\right]
=\mathcal L_{BB}+\mathcal L_{BF}+\Lff .
\end{equation}
Here $\beta_1,\beta_2$ are auxiliary odd factors, and the products and adjoints require a graded prescription. In particular,
\begin{equation}\label{eq:seed}
\Lff\propto\Tr(\beta_1\Psi_F^\dagger\beta_2\Psi_F)
\end{equation}
has a barred--unbarred matter pairing. Calling that pairing ``Hermitian type'' does not prove that the graded trace expression is self-adjoint: auxiliary adjoints and cyclic signs matter. The coefficientwise prescription in the book~\cite[Sec.~7.3.6]{Book} makes this distinction explicit. Neither the split in~\eqref{eq:Lgtd} nor the barred--unbarred notation assigns $\Psi_F$ to an $E_6$ representation by itself.

\section{The exceptional ledger in one conjugation convention}\label{sec:ledger}
We use complex representation branchings; compact forms specify positive invariant Hermitian pairings when needed. Other real forms require their own reality and positivity analysis. The notation $E_8\times\omega E_8$ records two branches and the split-complex grading $\omega^2=1$, $\omega\ne\pm1$. Algebraically $\mathfrak e_8\otimes\mathbb R[\omega]\simeq\mathfrak e_8\oplus\mathfrak e_8$; multiplication by $\omega$ is not a new kind of Lie-group product.

For one branch, the Lie-algebra branching is
\begin{equation}\label{eq:e8}
\mathbf{248}=(\mathbf8,\mathbf1)\oplus(\mathbf1,\mathbf{78})
\oplus(\mathbf3,\mathbf{27})\oplus(\bar{\mathbf3},\overline{\mathbf{27}}).
\end{equation}
The first entry denotes $SU(3)_{\st}$, also called $SU(3)_{\mathrm{st}}$ in the program. It is retained as a geometric label, not converted into the physical family multiplicity. Subgroup expressions in this section initially denote Lie-algebra or covering-group branchings; central quotients are addressed in Sec.~\ref{sec:global}.

Write the internal factors in the order $(F,c,W)$ and fix
\begin{align}
V=\mathbf{27}&=V_A\oplus V_B\oplus V_C,\label{eq:27}\\
V_A&=(\bar3,1,3),\quad V_B=(1,3,\bar3),\quad V_C=(3,\bar3,1),\nonumber\\
\mathbf{78}&=(8,1,1)\oplus(1,8,1)\oplus(1,1,8)
\oplus(3,3,3)\oplus(\bar3,\bar3,\bar3).\label{eq:78}
\end{align}
The bars in~\eqref{eq:78} are correlated with~\eqref{eq:27}. Mixing this convention with a $(3,3,\bar3)$ adjoint coset is inconsistent. The displayed pair follows from the KVS convention by conjugating the family and colour factors together. These are changes of representation convention, not changes of dimension or of physical family number. Standard trinification decompositions and explicit commutators are given in Ref.~\cite{BabuBajcSusic2023}. The subscripts on $V_A,V_B,V_C$ name comparison blocks and are distinct from the residual ledger labels $A$--$E$ below.

Choose $SU(2)_W$ on the first two entries of $3_W$, and an integer-normalized comparison Cartan
\begin{equation}\label{eq:qW}
q_W=\operatorname{diag}(1,1,-2),\qquad
3_W\downarrow SU(2)_W\times U(1)_W=2_{1}\oplus1_{-2}.
\end{equation}
The conjugate restricts to $2_{-1}\oplus1_{2}$. Table~\ref{tab:ledger} follows directly. Its $q_W$ labels are not assertions about physical hypercharge for the ledger entries.

\begin{table}[htbp]
\centering\small
\begin{tabular}{@{}p{0.18\textwidth}p{0.51\textwidth}r@{}}
\toprule
Selection & $(SU(3)_{\st},SU(3)_F,SU(3)_c,SU(2)_W)_{q_W}$ & Dimension\\
\midrule
Matched algebra & $(8,1,1,1)_0+(1,8,1,1)_0+(1,1,8,1)_0$ &24\\
Matched algebra & $(1,1,1,3)_0+(1,1,1,1)_0$ &4\\
Matched coset & $(1,1,1,2)_{3}+(1,1,1,2)_{-3}$ &4\\
Matched doublets & $(1,3,3,2)_{1}+\mathrm{c.c.}$ &36\\
Matched doublets & $(3,\bar3,1,2)_{1}+\mathrm{c.c.}$ &36\\
\midrule
Residual $A$ & $(1,3,3,1)_{-2}+\mathrm{c.c.}$ &18\\
Residual $B$ & $(3,3,\bar3,1)_0+\mathrm{c.c.}$ &54\\
Residual $C$ & $(3,\bar3,1,1)_{-2}+\mathrm{c.c.}$ &18\\
Residual $D$ & $(3,1,3,2)_{-1}+\mathrm{c.c.}$ &36\\
Residual $E$ & $(3,1,3,1)_{2}+\mathrm{c.c.}$ &18\\
\bottomrule
\end{tabular}
\caption{One-branch ledger obtained from Eqs.~\eqref{eq:e8}--\eqref{eq:qW}. Complex conjugation conjugates all factors and reverses $q_W$. The matched/residual selection is defined by the displayed rows and is a convention for comparing labels with a proposed physical roster. Dimensions are complex branching dimensions of the complexified adjoint, equal in total to the real dimension of its compact form. They are not counts of independent Weyl fields.}\label{tab:ledger}
\end{table}

The selection retains the $32$ algebra/coset dimensions and the two $36$-dimensional doublet pairs, giving $104$ matched and $144$ residual per branch. Two branches therefore give
\begin{equation}\label{eq:ledger}
2(32+72)=208,\qquad2(18+54+18+36+18)=288,\qquad208+288=496.
\end{equation}
This is a selected subspace decomposition after branching, not an $E_8$-invariant split into representations of dimensions $208$ and $288$. The criterion is a roster convention, not a derived particle-selection rule. For example, omitting the four comparison coset labels per branch would change the numerical partition to $200+296$, without changing the exceptional algebra. The particular numbers $208$ and $288$ must therefore always be accompanied by their selection convention. The matched doublet row with a geometric triplet cannot literally be physical leptons if physical leptons are geometric singlets. The same caution applies to every proposed field name attached to a row. The external geometric adjoint has dimension eight per branch, hence sixteen across both branches; this is already included in~\eqref{eq:ledger}.

The second branch has the analogous internal comparison factors, interpreted in the program as pre-gravitational. Identifying its colour, weak or family actions with visible-sector actions needs a specified interface. A duplication of labels does not establish such an interface, a four-dimensional Lorentz group, or gravitational dynamics. The proposed split-biquaternionic and $so(3,3)$ constructions address those separate questions~\cite{STM,WSI}.

\subsection{Why the quark-doublet labels occur in the adjoint}\label{sec:adjoint}
The proposed quark-doublet labels occur in the adjoint $78$, rather than in the fundamental $27$, in this accounting. In the convention of Table~\ref{tab:ledger},
\begin{equation}\label{eq:Qancestry}
(\mathbf1,\mathbf{78})\supset
(1,3,3,3)\longrightarrow
\underbrace{(1,3,3,2)_{1}}_{\text{quark-doublet labels}}
\oplus\underbrace{(1,3,3,1)_{-2}}_{\text{residual }A},
\end{equation}
with the conjugate multiplet also present. The doublet has $3\times3\times2=18$ complex branching dimensions, the nonabelian multiplicity needed for three family copies of coloured weak doublets. In contrast, the lepton-doublet labels arise from $V_A$ and its dual in the mixed sectors $(3_{\st},27)$ and $(\bar3_{\st},\overline{27})$. This is the ancestry of labels used in KVS~\cite{KVS}; the bars have been translated into the correlated convention of Eqs.~\eqref{eq:27}--\eqref{eq:78}.

Equation~\eqref{eq:Qancestry} also shows why selecting the doublet while leaving its singlet partner residual is not an $E_6$-invariant operation. It is a selection made after branching. No physical hypercharge has yet been assigned to either term. Where a physical colour triplet is associated with the conjugate row, the colour outer automorphism must be applied consistently to the whole colour factor, rather than separately to individual fields.

There is useful exceptional structure behind this assignment. On complexification, set
\begin{equation}\label{eq:grading}
\mathfrak e_6^{\mathbb C}=\mathfrak h\oplus X\oplus X^*,\qquad
\mathfrak h=\mathfrak{sl}_3^{\oplus3},\quad X=(3,3,3).
\end{equation}
The brackets satisfy $[\mathfrak h,X]\subset X$, $[X,X^*]\subset\mathfrak h$ and $[X,X]\subset X^*$. Thus the connecting generators carry all three indices in a single tensor, with their products constrained by a $\mathbb Z_3$ grading. This $27$-dimensional subgroup module $X$ is not the fundamental $27$ of $E_6$, and the order-three grading is not a count of physical generations. Explicit matrix checks of these relations are described in Appendix~\ref{app:verification}. The brackets can constrain a proposed field interaction once a physical action of the relevant generators has been constructed; they do not themselves provide a Yukawa interaction or localization mechanism.

An internal adjoint representation does not dictate Lorentz spin or statistics: adjoint-valued fermion fields are possible. Nevertheless, taking every component of the full self-conjugate $78$ as an independent left-handed Weyl coefficient would supply conjugate-paired internal content. The three quark-doublet multiplicities alone therefore do not select a chiral spectrum. Antiparticles of a Weyl field should not be confused with additional independent mirror Weyl fields. The physical benchmark below is specified independently of this adjoint assignment.

\subsection{A restricted obstruction to matter hypercharge}\label{sec:hypercharge}
The comparison Cartan $q_W$ and physical hypercharge $Y$ must not be equated without checking their actions. The following obstruction is useful because its assumptions are explicit.
\begin{proposition}[Single-branch hypercharge obstruction]\label{prop:Y}
Keep the trinification embedding~\eqref{eq:27}--\eqref{eq:78}, its standard $SU(2)_W$ subgroup, and a hypercharge generator lying in $\mathfrak e_6$ and commuting with all of $\mathfrak{su}(3)_F\oplus\mathfrak{su}(3)_c\oplus\mathfrak{su}(2)_W$. No such generator gives both the quark- and lepton-doublet types selected in Table~\ref{tab:ledger} their physical hypercharge magnitudes $1/6$ and $1/2$.
\end{proposition}
\begin{proof}
A commuting generator is a singlet in the adjoint restricted to the stated subgroup. The two family/colour octets in~\eqref{eq:78} have no such singlet, and neither does $X\oplus X^*$. The weak octet restricts to $3_0\oplus1_0\oplus2_3\oplus2_{-3}$, with exactly one singlet. Therefore the centralizer is one-dimensional, spanned by $q_W$, and $Y=a q_W$. Both selected matter-doublet types have $|q_W|=1$, including their conjugates. They would consequently have equal $|Y|=|a|$, whereas the physical magnitudes differ by a factor of three.
\end{proof}
This does not obstruct a charge-first physical module with $Y=Q-T_3$, nor prove a no-go for every cross-branch or extended construction. It does rule out identifying that module with the displayed ledger through this single commuting $E_6$ generator. A sector-dependent rescaling of $q_W$ is not one Lie-algebra generator. A cross-branch alternative must specify how both branches act on the same fields; writing a Cartan combination alone is insufficient. The scalar-current lift in Sec.~\ref{sec:bridge} uses a separate operator-level identification and is not a counterexample to Proposition~\ref{prop:Y}.

\section{Physical matter: a conditional charge-first benchmark}\label{sec:matter}
A minimal left ideal of $\Cl(6,\mathbb C)\simeq M_8(\mathbb C)$ has complex dimension eight. Three fermionic ladder pairs give number-operator multiplicities $1,3,3,1$ and the charge operator $N/3$~\cite{Furey}. This construction supplies useful charge and colour data. Eight internal states are not a complete chiral Standard-Model generation. In an all-left-handed Weyl convention a family requires fifteen gauge components, or sixteen with a singlet $\nu^c$. Lorentz spinor components and independent antiparticle copies are not included in these numbers.

For a transparent benchmark take two eight-state coefficient spaces
\begin{equation}\label{eq:S}
S_+=(\nu,d^c_1,d^c_2,d^c_3,u_1,u_2,u_3,e^c),\qquad
S_-=(\nu^c,d_1,d_2,d_3,u^c_1,u^c_2,u^c_3,e).
\end{equation}
These labels denote independent left-handed Weyl multiplet components, not a prescription identifying the two ideal blocks with Lorentz chirality. On $S_+\oplus S_-$ set
\begin{equation}\label{eq:Q}
Q=\tfrac13\operatorname{diag}(N,-N),\qquad
N=\operatorname{diag}(0,1,1,1,2,2,2,3).
\end{equation}
Let colour act as $3$ on $u,d$, as $\bar3$ on $u^c,d^c$, and trivially on leptons. Choose
\begin{equation}\label{eq:weak}
T_+=|\nu\rangle\langle e|+\sum_{a=1}^{3}|u_a\rangle\langle d_a|,
\quad T_-=T_+^\dagger,\quad T_3=\tfrac12[T_+,T_-],\quad Y=Q-T_3.
\end{equation}
Then $[T_3,T_\pm]=\pm T_\pm$, $[Y,T_\pm]=0$, and colour commutes with the weak action. The representation is
\begin{equation}\label{eq:SM}
\begin{split}
\mathcal F={}&(3,2)_{1/6}\oplus(\bar3,1)_{-2/3}\oplus(\bar3,1)_{1/3}\\
&\oplus(1,2)_{-1/2}\oplus(1,1)_1\oplus(1,1)_0.
\end{split}
\end{equation}
Thus $\mathcal H_{\rm phys}=\mathbb C^3\otimes\mathcal F$ is a $48$-dimensional internal coefficient space for three complete families with $\nu^c$. Here the tensor notation initially denotes a multiplicity space; it does not assert a uniform family-triplet action on all sixteen components. The common family action is specified in Sec.~\ref{sec:family}. This is a consistent \emph{chosen} benchmark, not a derivation of the weak action or the multiplicity from~\eqref{eq:Lgtd}. In particular, the off-block operators in~\eqref{eq:weak} do not arise merely from block-diagonal left multiplication by $\Cl(6)\oplus\Cl(6)$. Additional represented operations must be identified in a microscopic realization~\cite{Furey2018EPJC,LRbiquaternion}.

\subsection{Why charge-first data do not select chirality}
The alternative
\begin{equation}\label{eq:counter}
T_+^{\rm alt}=T_++|e^c\rangle\langle\nu^c|
+\sum_a|d^c_a\rangle\langle u^c_a|
\end{equation}
with the same construction of $T_-^{\rm alt},T_3^{\rm alt}$ and $Y^{\rm alt}=Q-T_3^{\rm alt}$ gives
\begin{equation}
(3,2)_{1/6}\oplus(\bar3,2)_{-1/6}\oplus(1,2)_{-1/2}\oplus(1,2)_{1/2}.
\end{equation}
It has exactly the same $Q$ and colour action as~\eqref{eq:Q}, but is vectorlike. Both assignments cancel local Standard-Model anomalies; both have an even number of weak doublets. Charge, colour and anomaly cancellation therefore do not uniquely choose~\eqref{eq:weak}. A localization or other selection mechanism must do that work. Writing $Y=Q-T_3$ is consistent once the weak action is fixed; it does not derive that action. Any proposed equality to a cross-branch Cartan combination must hold as an operator on the physical module and requires a branch-identification map.

For~\eqref{eq:SM}, the familiar anomaly sums per family are explicitly
\begin{align}
SU(3)_c^3:&\quad2-1-1=0,\nonumber\\
SU(3)_c^2Y:&\quad2(1/6)-2/3+1/3=0,\nonumber\\
SU(2)_L^2Y:&\quad3(1/6)-1/2=0,\nonumber\\
\Tr Y:&\quad6(1/6)+3(-2/3)+3(1/3)+2(-1/2)+1=0,\label{eq:anomaly}\\
\Tr Y^3:&\quad6(1/6)^3+3(-2/3)^3+3(1/3)^3+2(-1/2)^3+1=0.\nonumber
\end{align}
Common quadratic-index factors have been suppressed. There are four weak doublets per family and twelve for three families. These checks certify the chosen representation, not its emergence.

\subsection{An explicit common family action}\label{sec:family}
One possible common family action on the three copies of~\eqref{eq:SM} is stated in Table~\ref{tab:physical}. The superscript $c$ denotes the left-handed charge conjugate of a right-handed particle field, not an additional copy of its antiparticle. Physical left- and right-handed particle fields then both carry $\bar3_F$.
\begin{table}[htbp]
\centering
\begin{tabular}{@{}lccccc@{}}
\toprule
Left-handed field & $SU(3)_F$ & $SU(3)_c$ & $SU(2)_L$ & $Y$ & Components\\
\midrule
$Q$ & $\bar3$ & $3$ & $2$ & $1/6$ & $18$\\
$u^c$ & $3$ & $\bar3$ & $1$ & $-2/3$ & $9$\\
$d^c$ & $3$ & $\bar3$ & $1$ & $1/3$ & $9$\\
$L$ & $\bar3$ & $1$ & $2$ & $-1/2$ & $6$\\
$e^c$ & $3$ & $1$ & $1$ & $1$ & $3$\\
$\nu^c$ & $3$ & $1$ & $1$ & $0$ & $3$\\
\midrule
Total &&&&&$48$\\
\bottomrule
\end{tabular}
\caption{A chosen three-family coefficient representation. The counts exclude Lorentz spinor components and independent duplicate antiparticles. The table specifies the physical module, not a decomposition of the $208$ ledger labels.}\label{tab:physical}
\end{table}

If this common family group is gauged, the additional perturbative anomaly coefficients vanish. With cubic fundamental index $A(3)=1$ and quadratic index $T(3)=1/2$,
\begin{align}
\mathcal A_{F^3}&=-6+3+3-2+1+1=0,\label{eq:familyanomaly}\\
\mathcal A_{F^2Y}&=\tfrac12\left[6(1/6)+3(-2/3)+3(1/3)+2(-1/2)+1\right]=0.\nonumber
\end{align}
Mixed triangles with a single generator of a simple nonabelian factor vanish by tracelessness. Together with~\eqref{eq:anomaly}, these checks account for the perturbative gauge anomaly tensor of the product Lie algebra. There are twelve weak doublets, so the usual $SU(2)$ parity test also passes.

The last $+1$ in $\mathcal A_{F^3}$ comes from the single family triplet $\nu^c$, containing three SM-singlet neutrinos. Removing that triplet while keeping the other assignments fixed gives $\mathcal A_{F^3}=-1$. Because $Y(\nu^c)=0$, it contributes nothing to $\mathcal A_{F^2Y}$ or to any Standard-Model anomaly. Thus these neutrinos cancel the remaining \emph{cubic family} anomaly in this particular gauged-family benchmark, provided no other chiral fields compensate it. Standard-Model gauge consistency alone does not require them. This mechanism is familiar in gauged family models, for example Ref.~\cite{Hosek2016}; here it is an explicit consistency check, not a new general anomaly-cancellation principle. If $SU(3)_F$ is only global, its 't~Hooft anomaly has a different interpretation.

The neutrinos are sterile under the Standard Model but charged under $SU(3)_F$. Their mass terms must respect this distinction. In two-component notation the Lorentz-scalar product of identical left-handed fields is symmetric in family indices, so
\begin{equation}\label{eq:nuMajorana}
\nu^{c\,i}\nu^{c\,j}\ \in\ \operatorname{Sym}^2(3_F)=6_F.
\end{equation}
There is no family-singlet bare Majorana mass while $SU(3)_F$ is unbroken. Such a mass requires an effective coefficient transforming as $\bar6_F$, supplied, for example, by a scalar expectation value or a composite order parameter. This is a requirement on a possible mechanism, not a prediction of an elementary sextet. Dirac Yukawa couplings to a family-singlet Higgs are allowed, with family-degenerate coefficients until family symmetry is broken. Neither a seesaw scale nor a neutrino mass follows from~\eqref{eq:familyanomaly}.

The common family group is an input at the physical-module level. It has not been obtained by identifying the two exceptional branches' family groups. If gauged, its breaking, gauge-boson masses and flavour effects require analysis. Perturbative anomaly cancellation does not by itself settle the global gauge group. Section~\ref{sec:global} supplies one admissible global form for the benchmark and distinguishes it from the embedded trinification subgroup.

\subsection{Family multiplicity and the Albert comparison space}
A single $27$ cannot contain three complete families as a linear subrepresentation: their dimensions are $45$, or $48$ including $\nu^c$. This is a restriction on one representation, not on $E_6$ theories in general: three copies of a conventional matter $27$, for example, accommodate three families together with additional states. Relabelling a trinification factor as $SU(3)_F$ does not turn one $27$ into those three copies. Nor does $M_8(\mathbb C)$ have exactly three minimal ideals: there are continuously many choices of primitive idempotent, with isomorphic minimal modules. A selection of three copies is extra structure until a selection theorem or mode count is supplied.

The real Albert algebra obeys
\begin{equation}
\dim_{\mathbb R}J_3(\mathbb O)=3+3\cdot8=27.
\end{equation}
Its $24$ off-diagonal real coordinates do not constitute $24$ independent complex Weyl fields. After complexification, the $27$ used for $E_6$ comparison is irreducible; its diagonal/off-diagonal splitting relative to a chosen Jordan frame is not an $E_6$-invariant decomposition. The geometric and spectral use of this algebra~\cite{DrayManogue1999,Boyle2020,MassRatios} can be retained without identifying its off-diagonal coordinates with the complete physical spectrum. The numerical coincidence $24^2=576=2\cdot288$ has no established physical bilinear-count interpretation.

The no-go analysis of Distler and Garibaldi~\cite{DistlerGaribaldi} constrains attempts to obtain chiral matter from the $E_8$ adjoint under its specified embedding hypotheses. Keeping physical matter outside that adjoint avoids that particular identification; it does not prove the compatibility of a separate Clifford module with an exceptional gauge or geometric structure. That compatibility needs its own maps.

\section{Exceptional comparison operators and physical scalar currents}\label{sec:bridge}
For the complex comparison module $V=27$,
\begin{equation}\label{eq:end}
\End(V)\simeq\overline{27}\otimes27=1\oplus78\oplus650.
\end{equation}
This identity is exact. Three appearances of the number $78$ must nevertheless be distinguished: the adjoint summand in~\eqref{eq:e8}, the adjoint operator summand in~\eqref{eq:end}, and any physical current or scalar composite proposed to transform that way. Equality of representation type is not equality of carriers or operators. A current's Lorentz index, a scalar's contraction and a fermion's spinor index are independent of that internal label. Equation~\eqref{eq:end} is therefore not automatically the complete physical channel content of~\eqref{eq:seed}. Its use for physical matter requires a representation-compatible channel map. The $650$ is an internal representation; its larger dimension does not make its elements higher-engineering-dimension operators than the $78$ bilinears.

\subsection{The physical domain and comparison hypercharge}
Use the conditional common family action in Table~\ref{tab:physical}. It is an input, not a dynamically selected family symmetry. The scalar currents
\begin{align}
B_d^a{}_{ij}&=\bar d_{Rj}Q_{Li}^a,&B_e^a{}_{ij}&=\bar e_{Rj}L_{Li}^a,\nonumber\\
B_u^a{}_{ij}&=\bar u_{Rj}Q_{Li}^a,&B_\nu^a{}_{ij}&=\bar\nu_{Rj}L_{Li}^a\label{eq:currents}
\end{align}
are Lorentz scalars, with colour contracted. The first pair has $Y=+1/2$ and the second $Y=-1/2$. Each transforms in
\begin{equation}
\mathcal B_\pm=\End(\bar3_F)\otimes2_{\pm1/2}=(1_F\oplus8_F)\otimes2_{\pm1/2}.
\end{equation}
There are $18$ complex coefficient directions per species. Different species are distinct copies even when their quantum numbers coincide. When a comparison scalar is required to be Hermitian, its positive- and negative-hypercharge blocks are related by Hermitian conjugation (with the weak epsilon convention in Eq.~\eqref{eq:t}); their lift coefficients must obey the corresponding conjugation relation. The two signs alone do not imply two independent Higgs doublets.

On $3_W$ set $Y_W=q_W/6$ and $T_3^W=\operatorname{diag}(1/2,-1/2,0)$. With $E_{ab}$ the matrix units, the pairs
\begin{equation}\label{eq:t}
t^+=(E_{13},E_{23}),\qquad t^-=(E_{32},-E_{31})
\end{equation}
transform as weak doublets of hypercharge $+1/2$ and $-1/2$ under commutators. The corresponding electric charges are $(1,0)$ and $(0,-1)$. This matches \emph{operator} charges; it does not assign the comparison $27$ the charges of the physical fermions. In particular, setting $Y_W=q_W/6$ makes the weak-adjoint coset doublet have $|Y_W|=1/2$, but makes every selected ledger matter doublet have $|Y_W|=1/6$. The physical lepton doublet has $|Y|=1/2$. Proposition~\ref{prop:Y} explains why matching the scalar channel cannot supply the missing matter embedding.

\subsection{An explicit family of admissible lifts}\label{sec:lifts}
For either pair in~\eqref{eq:t}, write $B^a=b^a I_3+X^a$ with $b^a=\Tr B^a/3$ and $\Tr X^a=0$. In the ordered blocks $V_A,V_B,V_C$, define
\begin{align}
T(t)&=\operatorname{diag}(I_3\otimes t,-I_3\otimes t^T,0),\label{eq:TS}\\
S(t)&=\operatorname{diag}(I_3\otimes t,+I_3\otimes t^T,0).\nonumber
\end{align}
Here $T(t)$ is the full weak generator in the $78$. With the compact Hilbert--Schmidt pairing, $S(t)$ is traceless and orthogonal to the $78$, hence belongs to the $650$. To see the orthogonality, the two weak-generator traces cancel; family and colour generators give zero traces; the remaining $54$ adjoint generators are off-diagonal between the three blocks. The operator $\iota_A(X\otimes t)$, extended by zero outside $V_A$, belongs to $(8_F,1_c,8_W)\subset650$.

\begin{proposition}[Nonunique scalar-current lift]\label{prop:lift}
For every $\alpha\in\mathbb C$,
\begin{equation}\label{eq:lift}
\Phi_\alpha(B)=\sum_{a=1}^{2}\left\{
 b^a[\alpha T(t_a)+(1-\alpha)S(t_a)]+\iota_A(X^a\otimes t_a)\right\}
\end{equation}
is an injective intertwiner for the family, colour, weak and comparison-hypercharge actions on the stated channel spaces. Its $V_A$-block restriction is $\sum_a B^a\otimes t_a$.
\end{proposition}
\begin{proof}
The family action on $B$ is $B\mapsto U_F^* B U_F^T$, agreeing with conjugation on the $V_A$ block. The trace is invariant, and the traceless part transforms as the adjoint. Each $t_a$ has the weak and hypercharge transformation in~\eqref{eq:t}; the transposed block in~\eqref{eq:TS} transforms consistently in the dual weak representation. Colour acts trivially on these channels. Both $T$ and $S$ have $V_A$ restriction $I_3\otimes t$, so the stated readout follows and is injective.
\end{proof}
The choice $\alpha=1$ lifts the family singlet purely into the $78$, as in the target-lift construction~\cite{E6650}. The equally admissible $\alpha=1/2$ gives $\Phi_{1/2}(B)=\iota_A(\sum_a B^a\otimes t_a)$. Thus symmetry and readout do not select the pure-$78$ lift. In fact, for each hypercharge sign the colour-singlet weak-doublet family singlet has multiplicity two in $\End(V)$, while the family octet has multiplicity one. These come from $\End(V_A)$ and the colour trace in $\End(V_B)$; the other blocks supply no such colour-singlet doublets. Before fixing readout, there are three complex intertwiner coefficients per species.

Let $P_1 B=(\Tr B)I_3/3$ and $P_8=I-P_1$, acting separately on each weak component. With $\Tr(t_a^\dagger t_b)=\delta_{ab}$,
\begin{equation}\label{eq:gram}
\Phi_\alpha^\dagger\Phi_\alpha
=2\bigl(|\alpha|^2+|1-\alpha|^2\bigr)P_1+P_8.
\end{equation}
If one additionally postulates the comparison operator $G=g_{78}P_{78}+g_{650}P_{650}$, its pullback is
\begin{equation}\label{eq:Kalpha}
K_\alpha=2\bigl(g_{78}|\alpha|^2+g_{650}|1-\alpha|^2\bigr)P_1+g_{650}P_8.
\end{equation}
This is an algebraic calculation conditional on both $G$ and $\Phi_\alpha$. The factor two for $\alpha=1$ is not a prediction of a physical coupling ratio. At $\alpha=1/2$ the map is already an isometry and the singlet coefficient is $(g_{78}+g_{650})/2$. A subgroup-invariant response can also mix the two comparison singlets; full $E_6$ invariance of the response is an additional condition. These lifts do not establish an $E_6$ action on the physical $48$-state space. Independent species require species-labelled maps and interactions, not one injective map of all their currents into a single $18$-dimensional image.

\subsection{A physical global form and the comparison action}\label{sec:global}
Local anomaly sums concern Lie-algebra generators. A global gauge group must also act consistently through its central identifications. For~\eqref{eq:27}, the compact trinification subgroup inside simply connected $E_6$ is $(SU(3)_F\times SU(3)_c\times SU(3)_W)/\mathbb Z_3$, with diagonal centre acting trivially on $V$. Let $z=e^{2\pi i/3}$. On the covering product
\begin{equation}
\widetilde G=SU(3)_F\times SU(3)_c\times SU(2)_W\times U(1)_\theta,
\end{equation}
embed the last two factors in $SU(3)_W$ by
\begin{equation}\label{eq:globalembedding}
(U,e^{i\theta})\longmapsto\operatorname{diag}(e^{i\theta}U,e^{-2i\theta}),
\qquad q=6Y.
\end{equation}
The element
\begin{equation}\label{eq:kernelE6}
k_{E_6}=(zI_3,zI_3,-I_2,e^{-i\pi/3})
\end{equation}
lies in the kernel of the map to this subgroup of $E_6$. It acts on the physical $\nu^c\sim(3_F,1,1)_0$ as $z$. The full matter module therefore does not descend to that embedded subgroup. The selected currents do: their family centre is trivial and their weak and $U(1)$ phases cancel. This is a fermion-descent obstruction, not an obstruction to the current lift.

One can state a global form for the physical benchmark more constructively. Every multiplet in Table~\ref{tab:physical} is invariant under
\begin{equation}\label{eq:kernelSM}
k_{\rm SM}=(I_3,zI_3,-I_2,e^{i\pi/3}).
\end{equation}
Its powers give the entire central kernel on this matter module: $\nu^c$ forces the family-centre phase to be trivial, $e^c$ fixes $\theta$ to multiples of $\pi/3$, and the remaining multiplets fix the colour and weak phases. Thus a faithful group acting on these local matter fields is
\begin{equation}\label{eq:Gphys}
G_{\rm phys}=\widetilde G/\langle k_{\rm SM}\rangle
=SU(3)_F\times\frac{SU(3)_c\times SU(2)_W\times U(1)_\theta}{\mathbb Z_6}.
\end{equation}
This is an admissible choice; the local matter content alone does not select the global theory or its line operators~\cite{Tong2017}.

Under~\eqref{eq:globalembedding}, $k_{\rm SM}$ acts on each of $V_A,V_B,V_C$ by the same scalar $z^{-1}$. It therefore acts trivially by conjugation on \emph{all} of $\End(V)$, even though it does not act trivially on $V$. Consequently the comparison operator action, and in particular $\Phi_\alpha$, are well defined for $G_{\rm phys}$. Equivalently, the comparison action factors through the adjoint form $E_6/\mathbb Z_3$. Its map from $G_{\rm phys}$ has an additional central $\mathbb Z_3$ kernel, represented by $(zI_3,z^{-1}I_3,I_2,1)$, which need not act trivially on physical fermions. A well-defined common action on the currents is therefore available without identifying the matter representation with an $E_6$ representation. This resolves the global question for the displayed operator interface only; it neither evades Proposition~\ref{prop:Y} nor constructs an $E_6$ action on the $48$ matter coefficients.

\section{What can be excluded from a bilinear reservoir?}\label{sec:exclusion}
A useful restricted statement is that a barred--unbarred bilinear of $R=3_F\otimes U$, with $U$ family-neutral and no family-charged insertions, has family content in $1_F\oplus8_F$. The same holds with $3_F$ and $\bar3_F$ exchanged. That statement does not apply to the entire comparison $27$, whose family restriction is not homogeneous. It also must not be extended without checking the cross terms of a physical module containing both family representations, such as Table~\ref{tab:physical}. From~\eqref{eq:27},
\begin{equation}
V\big|_{SU(3)_F}=3\,\bar3_F\oplus9\,1_F\oplus3\,3_F,
\end{equation}
so $\End(V)$ includes cross-block fundamental family channels. For example,
\begin{equation}\label{eq:counterEnd}
\Hom(V_B,V_C)=(3_F,\bar3_c\otimes\bar3_c,3_W)
=(3,3,3)\oplus(3,\bar6,3).
\end{equation}
The representation type $(3,3,3)$ also occurs in the $78$ and contains the comparison type of residual $A$ after weak branching. This is a statement about representation types: a block-supported operator in $\Hom(V_B,V_C)$ need not itself be a complete adjoint generator, whose action can correlate several blocks. Thus a barred--unbarred comparison square cannot be declared disjoint from that residual type on the basis of the homogeneous-family lemma. This does not identify a physical composite $A$; it shows why the homogeneous-family selection rule cannot exclude its comparison type.

Charge-difference and charge-sum expressions remain different operator constructions. Their quantum numbers, however, must be computed from the actual constituent representations. A representation type can occur in different products. A Nambu notation $(\Psi,\Psi^c)$ organizes anomalous pairings but is not necessarily a new set of independent particles, nor is it the only way to write charge-sum operators. Lorentz contractions and Fermi statistics must be checked separately. No unconditional claim is made that a physical diquark realizes the ledger $A$: the latter's comparison Cartan is not a physical electric-charge assignment.

\begin{proposition}[Restricted geometric-singlet obstruction]\label{prop:geom}
Let $G_{\st}$ be an independent internal group acting trivially on every physical matter field and its conjugate. Then derivative-free bilinears constructed solely from those fields and $G_{\st}$-invariant contraction tensors are $G_{\st}$ singlets. The result holds for barred--unbarred and charge-sum bilinears, also across branches when the fields are trivial under both geometric factors.
\end{proposition}
\begin{proof}
Every factor and contraction carries the trivial action. Their tensor product and any invariant projection are therefore trivial representations. Charge conjugation and a change to Nambu notation do not alter this statement.
\end{proof}
Under this assumption the $126$ geometric-nonsinglet residual dimensions per branch, or $252$ across both branches, cannot be those undressed bilinears. The $36$ geometric-singlet residual dimensions are not excluded by this argument. Nor does the proposition exclude operators dressed with geometric fields, nontrivial background tensors or derivatives transforming under the geometric group. It does not apply to spacetime/Lorentz spinors by declaring them ``geometric singlets'': the independent group in its hypothesis must be distinguished from the physical frame group. Triviality under that independent group is an assumption of this matter interpretation, not a consequence of the words ``Clifford spinor.'' The same test applies to matched as well as residual rows.

For the proposed lepton identification there is a more direct linear test:
\begin{equation}\label{eq:geomHom}
\Hom_{SU(3)_{\st}}(3_{\st},1_{\st})=0.
\end{equation}
With other factors held fixed, a geometric-triplet ledger label cannot be mapped nontrivially to geometric-singlet physical matter by an ordinary equivariant linear map. A contraction with a suitably transforming geometric section can change this conclusion, but that section and its physical interpretation must be supplied. This identifies a missing localization datum without requiring the external $SU(3)$ to be reinterpreted as a family group.

Consequently, the defensible conclusion is that the $288$ is \emph{not compelled} to be a particle reservoir. It is not that no theory realizing the same scaffold could contain additional particles. Representation labels alone neither establish propagating states nor forbid them. No light-particle exclusion follows from the ledger or from calling the labels nonpropagating.

\section{The dynamical matching problem}\label{sec:dynamics}
The geometric operator construction becomes explanatory only if a common microscopic state selects its maps and couplings. Here ``state'' includes the localized background, its correlations and the representation of the physical field algebra; the formal GTD equilibrium ensemble is not automatically that state. The following audit identifies the necessary inputs without fitting them to the proposed mass ratios.

\subsection{The interaction supplied and the interaction required}
The coassociative inter-aikyon interaction displayed in the book~\cite[Eq.~(7.91)]{Book} has the form
\begin{equation}\label{eq:psi}
V_{\rm na}^{(4)}=\kappa\,\psi_{\alpha\beta\gamma\delta}
\Tr_M(q_i^\alpha q_j^\beta q_k^\gamma q_l^\delta),\qquad(i,j,k,l)\ \hbox{not all equal},
\end{equation}
with summation over octonionic indices. It supplies a $G_2$-invariant tensor on internal two-forms. Its imaginary-two-form eigenvalues $2$ and $-1$ do not identify family channels or establish attraction in a physical Lorentz-scalar channel. The book separately assumes an attractive localized interaction~\cite[Eq.~(7.75)]{Book}
\begin{equation}\label{eq:NJL}
S_H^{(4F)}=-G_H\int d^4x\sqrt{|g|}\,\mathcal O_H^\dagger\mathcal O_H,
\end{equation}
and states that its sign, strength and scale are not derived. Establishing~\eqref{eq:NJL} from~\eqref{eq:psi} requires more than a subgroup projection. The possible bound-state calculation belongs to the dynamical problem exemplified by NJL-type analyses~\cite{NJL1961,BardeenHillLindner1990}.

A linearized localization dictionary would provide coefficients schematically written as $q_i^\alpha=\sum_A E_i^{\alpha A}\chi_A+\cdots$, where $A$ includes physical species, family, gauge and spinor indices and the appropriate independent conjugate variables. Then a tree-level vertex would involve
\begin{equation}\label{eq:V}
\mathcal V_{ABCD}=\kappa\sum_{ijkl}\lambda_{ijkl}\,
\psi_{\alpha\beta\gamma\delta}
\Tr_M(E_i^{\alpha A}E_j^{\beta B}E_k^{\gamma C}E_l^{\delta D}).
\end{equation}
This is a schematic matching target, with spacetime kernels and measures suppressed, not an evaluated physical coupling. The interaction network $\lambda_{ijkl}$, normalized coefficients $E$, surviving auxiliary factors, graded ordering and adjoint prescription must be supplied. Fierz rearrangement and the projection onto~\eqref{eq:currents} then define the physical scalar-channel coefficient in
\begin{equation}\label{eq:physicalK}
\Gamma_{\rm loc}\supset-\int d^4x\sqrt{|g|}\,
B_{f,ij}^{a\dagger}K^{fg}_{ij\,a;kl\,b}B_{g,kl}^{b}.
\end{equation}
In the unbroken gauge phase, a constant $K^{fg}$ in~\eqref{eq:physicalK} can mix only currents of equal hypercharge; additional charge-sum structures require their own invariant contractions. Exchange and coarse-graining contributions may supplement~\eqref{eq:V}. The book also gives a torsion-induced axial-current contact term~\cite[Eq.~(8.49)]{Book}; its scalar projections and its stated non-minimal/self-dual completion must be considered in a calculation of the total $K$. A calculation of the pure coassociative contribution is not a calculation of every interaction in the program.

\subsection{Input audit and normalization}\label{sec:audit}
Table~\ref{tab:audit} separates supplied structure from missing data. These are not all independent adjustable parameters: a successful microscopic theory could determine them together. They are nevertheless required to make the present matching calculation well defined.
\begin{longtable}{@{}>{\raggedright\arraybackslash}p{0.20\textwidth}>{\raggedright\arraybackslash}p{0.34\textwidth}>{\raggedright\arraybackslash}p{0.39\textwidth}@{}}
\caption{Inputs for physical matching. Book references are to the supplied 2026 manuscript~\cite{Book}; the exceptional comparison construction is in~\cite{E6650}.}\label{tab:audit}\\
\toprule
Object & Supplied structure & Remaining calculation or specification\\\midrule
\endfirsthead
\toprule Object & Supplied structure & Remaining calculation or specification\\\midrule
\endhead
\bottomrule\endfoot
Microscopic vertex & Coassociative tensor and a candidate inter-aikyon vertex, book Sec.~7.10.2. & Interaction network; fermionic component; auxiliary and graded prescription; exchange propagator.\\[4pt]
Localization & Mode expansion and proposed mode assignment, book Eqs.~(7.68)--(7.69). & Normalized profiles with family, weak, colour and branch indices; spinor-carrier intertwiner and physical overlaps $E$.\\[4pt]
Chiral selection & Opposite-leaf target P2 and vortex-type mechanism, book Secs.~6.4.3, 8.7. & Boundary conditions, zero-mode count and removal or lifting of unwanted mirror partners.\\[4pt]
Current matching & Candidate scalar doublets, book Eq.~(7.74); the lifts of Sec.~\ref{sec:lifts}. & Physical Lorentz/Fierz projection, species weights and selection of the three comparison-channel coefficients.\\[4pt]
State and correlations & Formal regulated measure and canonical weight, book Eqs.~(7.37)--(7.39). & A normalizable interacting localized state, regulator prescription, symmetry-breaking orientation and correlations.\\[4pt]
Fermion normalization & Coefficient action with $a_F=c_0\eta^3$, book Eq.~(7.29). & Four-dimensional bilinear action and two-point kinetic matrices $Z_L,Z_R$.\\[4pt]
Scalar response & Spectral-action structure and scalar seeds, book Secs.~7.4, 7.6. & Composite response, kinetic/mass mixing, pole spectrum and selected light channel.\\[4pt]
Mass interface & Candidate Jordan spectra and occupied-chain operators. & Common-vacuum selection of occupied support, left/right frames, channel map and normalized physical mass operator.\\
\end{longtable}

A further issue precedes normalization. The mixed microscopic term is linear in the odd variable $\Psi_F$ (with an odd auxiliary coefficient), whereas the physical Dirac action is bilinear in physical fermion fields. The book explicitly leaves the necessary pairing or mixed-sector cumulant open~\cite[Sec.~7.6.1]{Book}. A linear localization map alone cannot supply the second fermion factor. Thus the Connes-time coefficient $a_F$ is not already a four-dimensional fermion residue.

The following are distinct targets: the four-fermion coefficient $K$ in~\eqref{eq:physicalK}; connected composite correlators $\langle\mathcal O_r(x)\mathcal O_s^\dagger(0)\rangle_c$, whose poles diagnose states; and the kinetic coefficients extracted from the two-point effective action. They cannot be equated without the corresponding effective-action matching. At a specified matching scale, with positive-definite Hermitian kinetic matrices and the convention $-\bar f_L M_f f_R+\mathrm{h.c.}$,
\begin{equation}\label{eq:massnorm}
M_{f,\rm phys}=Z_{f,L}^{-1/2}M_fZ_{f,R}^{-1/2}.
\end{equation}
The physical Dirac masses at that scale are the singular values of this normalized matrix. For a Yukawa coefficient of an unnormalized scalar, its scalar normalization enters as well; comparison with pole masses or other running-mass conventions requires the appropriate matching and evolution. Hilbert--Schmidt isometry of a comparison lift fixes neither $Z_{f,L/R}$ nor the scalar kinetic matrix. Consequently this paper neither invalidates a Jordan spectral identity nor proves its equality to measured fermion mass ratios. That equality requires the dynamical map and normalization.

\subsection{Light states, anomalies and conditional phenomenology}
If the family symmetry acts identically on physical left- and right-handed particle fields and the Higgs is a family singlet, an invariant Yukawa matrix is proportional to $I_3$. Nondegenerate masses require family-breaking data. An octet channel permits such data but does not select them. If family-octet scalar doublets propagate, their mass-basis couplings also require a flavour analysis~\cite{E6650}. Existence of a comparison channel is not evidence for a light elementary $650$.

A large compositeness scale does not prove that every bound-state pole is heavy. Suppression by $v^2/\Lambda^2$ is a conditional effective-field-theory estimate for specified decoupling dimension-six effects with controlled coefficients, not a bound on all deviations. Neither $\Lambda\sim M_{\rm Pl}$ nor the ledger establishes exactly two light doublets, excludes additional generations or forbids superpartners. A proposed second doublet and sterile neutrinos remain model choices to be subjected to dynamics and experiment, not consequences of $496-208$.

The anomaly checks in Eqs.~\eqref{eq:anomaly} and~\eqref{eq:familyanomaly} pertain to the stated physical benchmark. They are not a complete ultraviolet or global consistency proof. Section~\ref{sec:global} supplies a compatible global operator action while showing why the matter does not descend to the embedded trinification subgroup. Every additional genuinely gauged branch needs its own anomaly and spectrum analysis. Representation labels without chiral fields carry no fermion anomaly, but this does not make every proposed physical realization of them anomaly-free.

In particular, a proposed extra $U(1)_X$ needs all mixed anomalies checked. Vectorial charges on each Dirac species do not suffice in the unbroken electroweak theory. If quarks have a common charge $x_q$ and leptons a common charge $x_\ell$, the same on left- and right-handed particle fields, the per-family coefficient is
\begin{equation}\label{eq:extraU1}
\mathcal A_{SU(2)_L^2X}=\tfrac12(3x_q+x_\ell),
\end{equation}
which need not vanish; one must sum all families and any other chiral fields. A mass-dependent charge giving unequal charges to the two members of a weak doublet also cannot commute with unbroken $SU(2)_L$. No anomaly cancellation for an additional pre-gravitational or dark gauge sector is inferred from the checks performed here.

\section{Conclusion}\label{sec:conclusion}
The retained result is a consistent representation ledger with a conventional $208/288$ partition. It establishes neither three physical generations inside a single $27$ nor a closed particle spectrum. A chosen $48$-state Clifford-based benchmark correctly carries three complete chiral Standard-Model families with singlet neutrinos, while its weak action, family multiplicity and localization remain inputs to be derived.

Nonstandard trinification can still have a concrete geometric role: its comparison operators admit the physical scalar-current quantum numbers through explicit subgroup-equivariant lifts. The family-octet channel lies in the $650$, but the singlet lift is nonunique. This limits what normalization and coupling claims can be inferred from the channel assignment alone. The full comparison square contains fundamental family channels; the geometric-singlet obstruction holds only with its stated group and operator hypotheses. The adjoint origin of the quark-doublet labels supplies an algebraically constrained tensor structure, but its conjugate-paired content and the single-branch hypercharge obstruction prevent that observation from being a derivation of physical chirality.

The next task is to determine normalized localized mode maps, interaction matrix elements, composite correlations and kinetic matrices from one specified microscopic state. Until then, the exceptional geometry and the Clifford matter representation are conditionally compatible at the displayed operator level, but their physical unification and mass predictions are incomplete. The residual $288$ is bookkeeping whose physical interpretation cannot be settled by dimension counting.

\appendix
\section{Reproducible finite checks and their scope}\label{app:verification}
The ancillary script uses sparse matrices and exact rational arithmetic. It constructs the three-mode Clifford ladder operators, their colour action, the chosen weak action and its vectorlike alternative. A permutation matrix independently identifies the chosen $48$-state construction with the direct sum of the six textbook multiplet types in Table~\ref{tab:physical}. The script checks the full symmetric local anomaly tensor for a basis of the twenty physical product-algebra generators, not just the displayed anomaly sums. This certifies the chosen finite representation; it does not select it from GTD.

The comparison-space checks use the three blocks in~\eqref{eq:27}. Here is an explicit convention for the connecting generators, also fixing the support argument used in Sec.~\ref{sec:lifts}. Write $A_{ik}\in V_A$, $B_{jl}\in V_B$ and $C_{mn}\in V_C$, with indices running from $1$ to $3$. For a generator $X^{pqr}$ of $X=(3,3,3)$, choose
\begin{equation}\label{eq:Xaction}
\begin{split}
(X^{pqr}C)_{ir}&=\sum_m\epsilon_{pmi}C_{mq},\\
(X^{pqr}A)_{ql}&=\sum_k\epsilon_{rkl}A_{pk},\\
(X^{pqr}B)_{pn}&=\sum_j\epsilon_{qjn}B_{jr},
\end{split}
\end{equation}
with all other output entries zero. The outputs in these three lines lie in $V_A,V_B,V_C$, respectively; the notation specifies a single operator on their direct sum. Its Hermitian transpose supplies the dual generator. Together with the twenty-four block-diagonal subgroup generators, these yield a $78$-dimensional matrix realization. The script checks their Gram rank, transformation as $X\oplus X^*$ and closure under commutators. In particular it checks the grading relations~\eqref{eq:grading}, using the actual connecting matrices rather than only their dimensions.

These matrices also allow direct checks that the proposed $650$ lift pieces are orthogonal to every adjoint generator. The ancillary calculations verify the lift Gram matrices, nonuniqueness, adjoint-singlet centralizer, both ledger totals and the stated centre kernels. They also check the absence of a family-invariant neutrino Majorana tensor and the distinction between Standard-Model and cubic-family anomaly cancellation. Rational real matrix units span the complexified coefficient spaces; the complex-coefficient Gram formula follows by Hermitian extension. No positivity of a noncompact real form, localization spectrum, interacting vacuum, physical composite pole or measured mass ratio is computed. The exact identities and the physical emergence problem must remain distinct.

\section*{Acknowledgements}
It is a pleasure to thank Jose Isidro, Ashutosh Kotwal, Bishnu Gupta Teli and P.\ Samuel Wesley for helpful discussions.

\section*{Statements and declarations}
\subsection*{Author contribution and detailed use of generative AI}
Tejinder P.\ Singh is the sole author. The aikyon research program, its physical motivation and geometric interpretation originate in the author's work and, where cited, work with his collaborators. The author supplied the research material and formulated the questions studied here: completeness of the proposed chiral generations, the meaning of the $208/288$ accounting, the status of spectrum closure, and the possible connection between trinification geometry and Clifford matter. He challenged the assistant's successive interpretations, required explicit counterexamples and adversarial checks, and directed the scope and presentation of the paper. His role included research direction, physical interpretation and editorial decisions. This statement does not attribute every AI-generated intermediate derivation to unaided human calculation.

Generative AI was used substantially in the analytical and writing workflow. The author-supplied preparation record reports use of OpenAI and Anthropic systems. The analytical work, source preparation and verification described here used OpenAI's ChatGPT/Codex assistant during October 2026. A precise model snapshot for each interaction was not recorded, and no unverified model-version label is asserted. The assistant examined supplied material and relevant primary literature, proposed representation checks and counterexamples, developed and executed exact-arithmetic code, and drafted and reorganized substantial portions of the text. Its contributions included the conditional weak-action construction and vectorlike counterexample, the full-$27$ cross-block test, the scalar-current lifts and their nonuniqueness, the family anomaly calculation, the hypercharge centralizer argument, and the localization and normalization audit. The global-centre calculation, neutrino mass-tensor restriction and extra-abelian anomaly example were also developed and checked with AI assistance.

The book-level interaction audit uses equation references and content carried in the author-supplied research material. The book itself was not available for fresh inspection during the final source-bundle audit. The displayed interactions and the mathematical arguments are stated explicitly so that the conclusions do not rely solely on inaccessible equation references; the book citation identifies the manuscript version to which those references pertain.

At the author's request, the assistant performed repeated adversarial internal reviews, attempting to falsify the proposed claims through alternative representations, explicit generator matrices, channel multiplicities, central quotients and normalization tests. Claims were narrowed where those tests required it. This was an AI-assisted internal audit, not independent journal peer review or an external expert endorsement. The assistant also prepared the LaTeX and ancillary files, ran the finite verification calculations, and checked the compiled document. Agreement between formulas and code establishes neither a microscopic vacuum nor the physical realization of the chosen assumptions. No localization spectrum, bound-state simulation, physical mass fit or experimental dataset was generated.

AI is not an author and cannot assume scientific accountability. Scientific interpretation, acceptance of the text, verification of claims and sources, and the decision to submit remain the human author's responsibility. This disclosure does not claim that automated checking constitutes independent human validation.

Author declaration: All AI assisted output in this manuscript has been independently verified by the human author, who takes sole responsibility for the correctness and validity of its intellectual content.

\subsection*{Data and code availability}
The results are algebraic and conditional; no experimental dataset is used. A self-contained Python 3 script, \texttt{anc/verification.py}, is supplied with the source and supplementary material. It uses exact rational arithmetic to check the stated branch accounting, charge and weak assignments, the complete local anomaly tensor of the chosen product-algebra benchmark, centre phases, the explicit adjoint grading, the restricted centralizer and the channel Gram identities. The derivations and assumptions are stated in the text, and the script is not a substitute for their mathematical review.


\begin{thebibliography}{99}
\raggedright

\bibitem{KVS}
P.~Kaushik, V.~Vaibhav, and T.~P.~Singh, ``A proposed $E_8\times E_8$ kinematic scaffolding for the Standard Model with pre-gravitation,'' arXiv:2206.06911v4 (2026), doi:\href{https://doi.org/10.1515/zna-2026-0190}{10.1515/zna-2026-0190} Zeitschrift f\"ur Naturforschung A (2026). 


\bibitem{STM}
T.~P.~Singh, ``Spacetime and internal symmetry from split bioctonions and the two extra $SU(3)$'s of $E_8\times\omega E_8$,'' Preprints 2025, 2025100437, doi:\href{https://doi.org/10.20944/preprints202510.0437.v4}{10.20944/preprints202510.0437.v4}.

\bibitem{LRbiquaternion}
V.~Vaibhav and T.~P.~Singh, ``Left-right symmetric fermions and sterile neutrinos from complex split biquaternions and bioctonions,'' Adv.\ Appl.\ Clifford Algebras \textbf{33}, 32 (2023), arXiv:2108.01858 [hep-ph].

\bibitem{GTDemergence}
T.~P.~Singh, ``Towards deriving the Standard Model coupled to gravity from Generalized Trace Dynamics via the spectral action principle,'' Universe \textbf{12}, 205 (2026), doi:\href{https://doi.org/10.3390/universe12070205}{10.3390/universe12070205}.

\bibitem{MassRatios}
T.~P.~Singh, ``Fermion mass ratios from the exceptional Jordan algebra,'' arXiv:2508.10131v7 [hep-ph] (2026).

\bibitem{Baez2002}
J.~C.~Baez, ``The octonions,'' Bull.\ Amer.\ Math.\ Soc.\ \textbf{39}, 145--205 (2002), arXiv:math/0105155 [math.RA].

\bibitem{Furey}
C.~Furey, ``Charge quantization from a number operator,'' Phys.\ Lett.\ B \textbf{742}, 195--199 (2015).

\bibitem{Furey2018EPJC}
C.~Furey, ``$SU(3)_c\times SU(2)_L\times U(1)_Y(\times U(1)_X)$ as a symmetry of division algebra ladder operators,'' Eur.\ Phys.\ J.\ C \textbf{78}, 375 (2018), arXiv:1806.00612 [hep-th].

\bibitem{Stoica2018}
O.~C.~Stoica, ``Leptons, quarks, and gauge from the complex Clifford algebra $\mathbb{C}\ell_6$,'' Adv.\ Appl.\ Clifford Algebras \textbf{28}, 52 (2018), arXiv:1702.04336 [hep-th].

\bibitem{Boyle2020}
L.~Boyle, ``The Standard Model, the exceptional Jordan algebra, and triality,'' arXiv:2006.16265 [hep-th] (2020).


\bibitem{BabuBajcSusic2023}
K.~S.~Babu, B.~Bajc, and V.~Susi\v c, ``Trinification from $E_6$ symmetry breaking,'' JHEP \textbf{07}, 011 (2023), arXiv:2305.16398 [hep-ph], doi:\href{https://doi.org/10.1007/JHEP07(2023)011}{10.1007/JHEP07(2023)011}.

\bibitem{Book}
T.~P.~Singh, \emph{Gravitation and Quantum Theory as Emergent Phenomena: The Aikyon Framework: A Programme for Unification}, book manuscript (2026), author-supplied version underlying the interaction audit; especially Secs.~6.4.3, 7.3.6, 7.4, 7.6, 7.10.2, 8.7 and 9.6. Equation numbers in the text refer to this manuscript version.

\bibitem{WSI}
P.~S.~Wesley, T.~P.~Singh, and J.~M.~Isidro, ``Gravity and electroweak sector from symmetry breaking of an $SO(3,3)$ BF theory,'' arXiv:2602.19151 [hep-th] (2026).

\bibitem{DrayManogue1999}
T.~Dray and C.~A.~Manogue, ``The exceptional Jordan eigenvalue problem,'' Int.\ J.\ Theor.\ Phys.\ \textbf{38}, 2901--2916 (1999), arXiv:math-ph/9910004.

\bibitem{DistlerGaribaldi}
J.~Distler and S.~Garibaldi, ``There is no `Theory of Everything' inside $E_8$,'' Commun.\ Math.\ Phys.\ \textbf{298}, 419--436 (2010), arXiv:0905.2658 [math.RT].

\bibitem{E6650}
T.~P.~Singh, ``Family-dependent Yukawa channels in an $E_8\times\omega E_8$ framework: An explicit $E_6$ 650 projector and a minimal-bridge obstruction,'' Zenodo (2026), doi:\href{https://doi.org/10.5281/zenodo.22013061}{10.5281/zenodo.22013061}.

\bibitem{Hosek2016}
J.~Ho\v sek, ``A model that underlies the Standard model,'' arXiv:1606.03292v2 [hep-ph] (2016).

\bibitem{Tong2017}
D.~Tong, ``Line operators in the Standard Model,'' JHEP \textbf{07}, 104 (2017), arXiv:1705.01853 [hep-th], doi:\href{https://doi.org/10.1007/JHEP07(2017)104}{10.1007/JHEP07(2017)104}.

\bibitem{NJL1961}
Y.~Nambu and G.~Jona-Lasinio, ``Dynamical model of elementary particles based on an analogy with superconductivity. I,'' Phys.\ Rev.\ \textbf{122}, 345--358 (1961).

\bibitem{BardeenHillLindner1990}
W.~A.~Bardeen, C.~T.~Hill, and M.~Lindner, ``Minimal dynamical symmetry breaking of the standard model,'' Phys.\ Rev.\ D \textbf{41}, 1647--1660 (1990).

\end{thebibliography}
\end{document}